\documentclass{article}
\usepackage[utf8]{inputenc}
\usepackage[T1]{fontenc}
\usepackage{amsmath, amssymb, algorithm, algpseudocode}
\usepackage{dsfont}
\usepackage{fancyhdr}
\usepackage[table,xcdraw]{xcolor}
\usepackage{graphicx}
\usepackage{multirow}
\usepackage{endnotes}
\usepackage{float}
\usepackage{subfigure}
\usepackage{import}
\usepackage{stackrel}
\usepackage{comment}
\usepackage{wrapfig}
\usepackage{amsthm}
\usepackage{rotating}
\usepackage[title]{appendix}
\usepackage{listings}
\usepackage{tcolorbox}
\usepackage{soul}
\usepackage{fullpage}
\usepackage{stmaryrd}
\usepackage{tikz}
\renewcommand\labelenumi{(\roman{enumi})}
\renewcommand\theenumi\labelenumi  
     
\usepackage{mleftright}
\mleftright
\newcommand{\qa}{\alpha}
\newcommand{\qb}{\beta}

\newcommand{\qd}{\delta}
\newcommand{\qD}{\Delta}
\newcommand{\qh}{\eta}

\newcommand{\qj}{\psi}
\newcommand{\ql}{\lambda}
\newcommand{\qL}{\Lambda}
\newcommand{\qe}{\varepsilon}
\newcommand{\qf}{\varphi}
\newcommand{\qF}{\Phi}
\newcommand{\qr}{\rho}
\newcommand{\qs}{\sigma}

\newcommand{\qt}{\tau}
\newcommand{\qo}{\omega}
\newcommand{\cA}{{\cal A}}
\newcommand{\cB}{{\cal B}}

\newcommand{\cD}{{\cal D}}

\newcommand{\cG}{{\cal G}}
\newcommand{\cH}{{\cal H}}
\newcommand{\cI}{{\cal I}}

\newcommand{\cK}{{\cal K}}
\newcommand{\cL}{{\cal L}}
\newcommand{\cR}{{\cal R}}
\newcommand{\cU}{{\cal U}}
\newcommand{\cV}{{\cal V}}
\newcommand{\RR}{\mathbb R}
\newcommand{\cX}{{\cal X}}

\newcommand{\Tr}{\mathrm{Tr}}
\newcommand{\bits}{\{0,1\}}
\newcommand{\be}{\begin{equation}}
\newcommand{\ee}{\end{equation}}
\newcommand{\bea}{\begin{eqnarray}}
\newcommand{\eea}{\end{eqnarray}}
\newcommand{\rd}{{\rm d}}

\newcommand{\ket}[1]{\left|#1\right\rangle}		
\newcommand{\bra}[1]{\left\langle#1\right|}

\newcommand{\ketbra}[2]{\left|#1\rangle\langle#2\right|}
\newcommand{\braket}[2]{\left\langle #1\lvert#2\right\rangle}

\newcommand{\eps}{\varepsilon}
\newcommand{\EE}{{\mathbb E}}

\newcommand{\ray}[1]{{\it\color{red}#1}}
\newcommand{\boris}[1]{{\it\color{blue}#1}}
\definecolor{darkred}{RGB}{179, 16, 32}

\usepackage[colorlinks=true, citecolor=blue, linkcolor=blue, urlcolor=blue]{hyperref}
\usepackage{bbm}
\usepackage{slashed}
\usepackage{enumitem}
\usepackage{authblk}

\usepackage[normalem]{ulem}
\usepackage{amsmath}

\usepackage{etoolbox} % for patching
\usepackage{amsthm}

\theoremstyle{plain}
\newtheorem{theorem}{Theorem}[section]

\newtheorem{prop}[theorem]{Proposition}
\newtheorem{definition}[theorem]{Definition}

\newtheorem{lemma}[theorem]{Lemma}
\newtheorem{corollary}[theorem]{Corollary}

\usepackage[backend=bibtex, style=numeric,   maxnames=999, minnames=999, sorting=none, giveninits=true]{biblatex}
\begin{document}
	
\title{Continuous-variable approximate unitary 2-design, \\with applications to unclonable encryption}
\author{Arpan Akash Ray and Boris \v{S}kori\'{c}}
\affil{Eindhoven University of Technology, The Netherlands }
\date{ } 
\maketitle

\begin{abstract}
\noindent
We introduce an $\qe$-approximate unitary 2-design that is compatible with the
structure of p- and q-quadratures in continuous-variable (CV) quantum systems.
The design unitaries are defined on a finite-dimensional discretisation of the CV space
and can be physically implemented as operations on the full CV space.
This establishes the first approximate unitary design for CV systems.
The design alternatingly acts with unitaries based on the quadrature operators $\hat q$ and~$\hat p$.
We prove that the parameter $\qe$ is given by $1/d^\ell$, where 
$d$ is the dimension of the truncated Hilbert space and $\ell$ is the number of iterations.

We propose an Unclonable Encryption scheme in which the encryption operators are given by the 
unitaries which constitute the approximate unitary design. 
We prove its security using recent results on decoupling.
This establishes unclonable-indistinguishable security for a CV encryption for the first time.
\end{abstract}

%===============================================================
\section{Introduction}

%------------------------------------------------------------------
\subsection{Unitary two-designs}

Random unitaries drawn from the Haar measure on the unitary group play a central role in quantum information theory, because they provide a mathematically clean notion of ``generic'' dynamics and enable uniform averages that underlie decoupling, benchmarking, tomography, and cryptographic reductions.  
In practice, however, Haar-random unitaries are rarely available as a physical primitive, and even sampling them exactly is computationally expensive in high dimensions \cite{Mezzadri2007}.  
A standard remedy is to replace Haar randomness by \emph{unitary designs}: structured ensembles of unitaries that reproduce Haar statistics up to some fixed polynomial degree.  
For many protocols, matching only the first few moments already suffices, so low-order designs provide a bridge between idealized Haar arguments and implementable randomization schemes~\cite{DCEL2009}.
One convenient way to define unitary designs is via the $t$-fold twirling channel $\cG^{(t)}_\mu$
induced by a probability measure $\mu$ on $\mathrm U(d)$, the space of unitary operators in $d$ dimensions,
\begin{equation}
\label{eq:twirl-def}
\mathcal G^{(t)}_{\mu}(X)
=
\int_{\mathrm U(d)} U^{\otimes t}\,X\,(U^\dagger)^{\otimes t}\,\rd\mu(U),
\qquad
X\in\mathcal{L}\!\bigl((\mathbb C^{d})^{\otimes t}\bigr).
\end{equation}
Here $\cL$ stands for the space of linear operators.
The measure $\mu$ is called an (exact) \emph{unitary $t$-design} if $\mathcal G^{(t)}_{\mu}$ coincides with the Haar twirl $\mathcal G^{(t)}_{\mathrm{Haar}}$ (equivalently, if all balanced polynomials of degree at most $t$ in the matrix entries of $U$ have the same average under $\mu$ as under Haar) \cite{DCEL2009,CLLW2016}.  
The case $t=2$ is especially prominent: it captures second-moment statistics and is therefore the minimal notion of pseudorandomness that already supports a wide range of `average-case' tasks.  
It also underlies optimal strategies for state and process tomography and related estimation problems 
\cite{CIBA2025,BlumeKohoutTurner2014Curious,DCEL2009}.

Unitary designs exist for finite dimension~$d$.
A natural question is whether analogous objects exist for Continuous-Variable (CV) quantum systems.  
Here one immediately encounters conceptual and technical obstacles: CV systems are modeled on 
separable infinite-dimensional Hilbert spaces (e.g.\ $L^2(\mathbb R)$ for a single bosonic mode), and there is no direct counterpart of the finite-dimensional Haar measure that yields finite, well-behaved moment operators on the full unitary group.  Motivated by physical implementability, a particularly important candidate class is the set of Gaussian states/unitaries, often viewed as the CV analogue of the Clifford 
group.
However, Blume-Kohout and Turner \cite{BlumeKohoutTurner2014Curious}
showed that Gaussian states \emph{cannot} form a CV 2-design on $L^2(\mathbb R)$, 
and moreover that no true Gaussian ensemble on the entire infinite CV space is even ``close'' to a 2-design, because the relevant invariant averaging fails to produce a normalizable density operator.  
This already rules out the most accessible CV candidate family for reproducing second-moment Haar statistics.

Another perspective on CV `design-like' randomness was developed by Zhuang \emph{et al.} \cite{ZSYY2019}
in the study of scrambling in phase space.
Because phase space has infinite volume, they emphasize that even CV analogues of uniform ensembles (e.g.\ displacements, which play the role of Pauli operators) are not normalizable, so some form of regularization is unavoidable.
They proposed several methods to analyse higher-moment pseudorandomness 
on random Gaussian circuits, finding that large squeezing can produce signatures reminiscent of higher-design behavior while still falling short of an exact CV 2-design.
This shows the need for explicit energy or discretisation constraints in any CV notion of (approximate) unitary 2-design, and 
connects
with later general non-existence results for exact CV designs of order $t\ge 2$ which we describe below.

Finally, Iosue \emph{et al.}\ \cite{ISGA2024} recently published a systematic study of designs in infinite-dimensional settings and proved a much stronger no-go theorem: under the natural CV extensions of finite-dimensional design definitions (including those used in \cite{BlumeKohoutTurner2014Curious}), \emph{no} CV state $t$-designs exist for any $t\ge 2$, and likewise \emph{no} CV unitary $t$-designs exist for any $t\ge 2$ on separable infinite-dimensional Hilbert spaces. 
They introduced an alternative framework for \emph{state} designs based on rigged Hilbert spaces (``rigged $t$-designs'') and gave explicit constructions for $t=2$, together with regularization procedures that turn these into approximate, physically normalizable ensembles under energy cutoffs.  
For \emph{unitary} designs, they also discuss regularized notions of approximation, but explicit constructions beyond $t=1$ remain an open problem.

Given these no-go results, it is natural to relax the exact equalities and consider \emph{approximate} designs.  
In finite dimension, an $\varepsilon$-approximate unitary 2-design can be defined by requiring the induced 2-fold twirl to be 
$\qe$-close to the Haar twirl in a suitable norm,
%$\bigl\|\mathcal G^{(2)}_{\mu}-\mathcal G^{(2)}_{\mathrm{Haar}}\bigr\|_{\diamond}\le \varepsilon$,
with several equivalent or comparable formulations in terms of moment operators, induced Schatten norms, or frame potentials \cite{DCEL2009,Harrow2009}.  
Approximate 2-designs are abundant in discrete-variable systems and admit efficient constructions; in fact, much of the usefulness of designs in applications such as fidelity estimation persists under quantitatively controlled approximation \cite{DCEL2009}.
A conceptually simple discrete-variable construction relevant for the current work is 
due to Nakata \emph{et al.}\,\cite{Nakata2017Unitary2Designs}, 
who showed that alternating random unitaries diagonal in the Pauli-$Z$ basis with random unitaries 
diagonal in the Pauli-$X$ basis yield efficient $\varepsilon$-approximate unitary 2-designs, with the approximation error decaying as $\qe=\Theta(d^{-\ell})$ after $\ell$ repetitions.

In contrast to the maturity of the finite-dimensional case, there is currently no 
explicit construction of an $\varepsilon$-approximate \emph{CV unitary} 2-design.  

%------------------------------------------------------------------
\subsection{Decoupling and Unclonable Encryption}

Unclonable Encryption (UE) is a form of encryption of classical messages 
in which the ciphertext is quantum and cannot be split into two parts that can both be decrypted successfully once the decryption key is revealed. 
Broadbent and Lord \cite{broadbent_et_al:LIPIcs.TQC.2020.4} 
formalized
two security notions: \emph{unclonable security} and \emph{unclonable indistinguishability}.
Security is defined via a cloning-style experiment: an adversary A receives a challenge ciphertext, produces two quantum registers intended for two non-communicating parties B and C, and after the key is revealed, both parties attempt to recover the correct plaintext. 
Unclonable security demands that the probability that \emph{both} parties output the correct message is small for any attack.\footnote{
An earlier work by Gottesman \cite{Got23} introduced an `unclonable encryption' notion somewhat similar to unclonable security. 
}
Unclonable indistinguishability strengthens this notion by 
having only two possible plaintexts
and requiring that it is infeasible for \emph{both} B and C 
to correctly identify the plaintext choice with non-negligible advantage.
In the case of a single-bit plaintext the two notions are identical.

Majenz et al.\,\cite{MST2021} characterised an optimal UE format for the case of uniformly distributed plaintext:
the plaintext $m$ gets encoded as a fully mixed state on a subspace belonging to $m$, and then
gets mixed into a larger Hilbert space by applying a unitary (the key) that is drawn from the Haar measure or from a 
unitary two-design.

The original UE construction of Broadbent and Lord achieves unclonable indistinguishable security
in the quantum random oracle model. 
Recently, Bhattacharyya and Culf \cite{bhattacharyyaCulf2026uncloneable}
showed that one-bit unclonable indistinguishability can be achieved without computational assumptions, and
established information-theoretic unclonable indistinguishability based on {\em decoupling}.

Decoupling is a randomisation mechanism that destroys correlations between a target system and a reference system.
One-shot decoupling theorems \cite{dupuis2014oneshotdecoupling} give sufficient (and essentially tight) 
entropic criteria guaranteeing bounds on the decoupling. 
Decoupling is closely tied to state merging and the operational meaning of conditional entropy developed by Horodecki, Oppenheim, and Winter \cite{horodecki2005partial,horodecki2007merging}.
While many decoupling statements are first proven for Haar-random unitaries, only second-moment information is used, and hence unitary $2$-designs suffice. Moreover, approximate $2$-designs yield approximate decoupling guarantees. 
This connection was made explicit by Szehr \emph{et al.} \cite{szehr2013decoupling}.

The first continuous-variable unclonable encryption  scheme was proposed by Ray and \v{S}kori\'c \cite{raySkoric2025cvue}. 
Their construction achieves {unclonable security} in the sense of Broadbent--Lord, 
but not {unclonable indistinguishability}. 
Lessons from Bhattacharyya and Culf suggest a natural CV route: 
if a physically meaningful \emph{approximate} CV unitary $2$-design exists
 (necessarily with a regularization), 
then its unitaries can be used as encryption operators 
to achieve unclonable indistinguishability, with a security proof based on decoupling.
Since no such design exists in the literature, this path remains unexplored.

%------------------------------------------------------------------
\subsection{Contributions and outline}
\label{sec:contrib}

Inspired by the construction of
Nakata \emph{et al.}\,\cite{Nakata2017Unitary2Designs},
we explore the idea of back-and-forth projections in a different setting.
Instead of operators acting on an $n$-qubit Hilbert space,
we consider operators (on a finite-dimensional Hilbert space) that are compatible with 
the non-commuting CV quadrature operators $\hat p,\hat q$, 
in the sense that the unitaries act on a discretised version of the single-mode CV Hilbert space.
Our contributions are summarised as follows.

\begin{itemize}[leftmargin=4mm,itemsep=0mm]
\item
We introduce a discretisation of the CV one-mode Hilbert space,
by discretising the q-quadrature space into $d$ `tiles' of size~$\qD$.
This effectively replaces a state $\qr(q,q')$ by a piecewise constant approximation.
Furthermore we impose a cutoff $q_{\rm max}\propto 1/\qD$.
We show that the error $\| \qr_{\rm discr}-\qr \|_{\rm tr}$
induced by discretising a state in this way is upper bounded by
$\sqrt{\frac4\pi(\bar n+\frac12)/d}$, where $\bar n$ is the average photon number.
\item
On the $d$-dimensional Hilbert space we introduce unitaries of the form
$V_{\qa\qb}=\qo^{\qa\hat Q+\qb \hat Q^2}$
and $\tilde V_{\qa\qb}=\qo^{\qa\hat P+\qb \hat P^2}$, with uniformly drawn $\qa,\qb\in[0,d)$. 
Here $\qo=\exp(i\frac{2\pi}d)$;
the discrete quadrature $\hat Q$ acts on basis states $\ket j$ as $\hat Q\ket j=j \ket j$; 
the $\hat P$ is the Fourier transform of $\hat Q$ and represents the p-quadrature. 
We do two-fold twirls with the sequence $V\tilde VV$, in analogy with \cite{Nakata2017Unitary2Designs}.
We show that $\ell$ iterations of this twirl result in an $\qe$-approximate unitary two-design
with $\qe=(\frac1d)^\ell$.
This establishes for the first time an approximate unitary two-design that can be 
implemented in a CV Hilbert space.
\item
We introduce a variant of the scheme that works with integer parameters $\qa,\qb$,
i.e.\;a finite set of unitaries.
It performs exactly the same as the scheme with continuous $\qa,\qb$ but it needs
the dimension $d$ to be prime.
\item
Using the unitaries from the approximate two-design as encryption operators,
we construct an Unclonable Encryption scheme that encrypts a classical bit into 
the discretised CV space.
We provide a security proof in the Broadbent \& Lord famework.
This establishes CV unclonable indistinguishibility for the first time.
\end{itemize}
The outline of the paper is as follows.
In Section~\ref{sec:prelim} we briefly present preliminaries on norms,
unitary designs, and continuous variables.
In Section~\ref{sec:ourdesign} we introduce the discretisation and 
present our approximate unitary two-design.
Section~\ref{sec:photonicCV} discusses the implementability of our design
given the state of the art in photonic processing. 
In Section~\ref{sec:UE} we present the Unclonable Encryption scheme.
Section~\ref{sec:discussion} discusses limitations and future work.

%===============================================================    
\section{Preliminaries}
\label{sec:prelim}

%-----------------------------------------------------------------
\subsection{Hilbert spaces; operators; norms}

We use Dirac notation for quantum states.
The $k\times k$ identity matrix is written as $I_k$.
The notation $\cH_d$ stands for a Hilbert space of dimension~$d$. 
The set of linear operators on a vector space $\cH$ is denoted as $\cL(\cH)$;
the set of bounded operators is $\cB(\cH).$
The set of density matrices on a Hilbert space $\cH$ is written as $\cD(\cH)$.
A state $\qr\in\cD(\cH)$ satisfies $\qr\geq 0$, $\qr^\dagger=\qr$, and $\Tr\, \qr=1$.
The fully mixed state of a register $A$ is denoted as~$\qt^A$. 
The Hilbert-Schmidt inner product on operators is $\langle X,Y \rangle=\Tr\, X^\dagger Y$.
The Schatten $p$-norm ($p\geq 1$) of an operator $M$ is $\| M \|_p=(\Tr |M|^p)^{1/p}$, with $|M|=\sqrt{M^\dagger M}$.
The trace norm is $\| M \|_{\rm tr} = \frac12\| M \|_1$.
A superoperator maps operators to operators, potentially on a different space.
The most important superoperators for quantum physics are the completely positive trace-preserving (CPTP) maps.
A CPTP map $\qF: \cB(\cH_A)\to \cB(\cH_B)$ is a linear map
that satisfies $\Tr\, \qF(\qr)=\Tr\,\qr$ and 
$\forall_{k\in{\mathbb N}}\forall_{\qr\geq0} (\qF\otimes {\rm id}_k)(\qr)\geq 0$,
where ${\rm id}_k$ is the identity map on $k$-dimensional Hilbert space.
All quantum channels are CPTP maps.
Several norms exist for superoperators.
The `$q\to p$' norm is defined as 
$\| \qF \|_{q\to p}=\sup_{X\neq 0}\frac{\| \qF(X) \|_p}{\| X \|_q}$.
The diamond norm is defined as
$\| \qF \|_\diamond=\sup_k\| \qF\otimes {\rm id}_k \|_{1\to 1}$.

The following lemmas will be used in later sections.
 
\begin{lemma}
\label{lemma:1 norm vs 2 norm}
Let $X\in\mathbb{C}^{d_1\times d_2}$. It holds that
$ \|X\|_1 \leq    \sqrt{\min\{d_1,d_2\}}\,\|X\|_2 $. 
\end{lemma}

\begin{lemma}[Diamond vs $2\!\to\!2$ norm \cite{Low2010PseudorandonmessAL}]
\label{lem:dia-2}
Let $\qF: \cB(\cH_d)\to \cB(\cH_d)$ be a superoperator.
Then $\|\qF\|_\diamond\leq d\| \qF \|_{2\to2}$.
\end{lemma}

\begin{lemma}[Gentle Measurement Lemma \cite{Winter1999,Aar2004}]
\label{lem:gentle-m}
Let $\qL$ be a measurement operator.
Consider a state $\rho$ such that $\Tr \, \qL \qr \geq 1-\qe$.
Then the post-measurement state
$
\rho' \equiv \frac{\sqrt{\Lambda}\,\rho\,\sqrt{\Lambda}}{\operatorname{Tr}\{\Lambda \rho\}}$
is $\sqrt{\varepsilon}$-close to the original state $\rho$ in trace distance:
$ \|\rho-\rho'\|_{\rm tr} \leq \sqrt{\varepsilon} $.
\end{lemma}

%------------------------------------------------------
\subsection{Structure of approximate unitary 2-designs}
\label{sec:prelimTwoDesigns}

\begin{definition}[Approximate unitary 2-design]
Let $\mu$ be a measure on U$(d)$, the space of unitary operators acting on~$\cH_d$.
Let $\cG_\mu^{(2)}$ be the two-fold twirl map as defined by (\ref{eq:twirl-def}), with $t\!=\!2$,
and let $\cG_{\rm Haar}^{(2)}$ be the corresponding map using the Haar measure.
The measure $\mu$ is called an $\qe$-approximate 2-design if
\be
    \Big\| \cG_\mu^{(2)} - \cG_{\rm Haar}^{(2)} \Big\|_{2\to 2} \leq \qe.
\ee
\end{definition}
The construction of Nakata \emph{et al.}\,\cite{Nakata2017Unitary2Designs} uses
two families of unitaries: $U_Z(\qf)$, diagonal in the computational basis,
and $U_X(\chi)$, diagonal in the Hadamard basis.
The parameters $\qf,\chi$  are uniformly sampled from (some discrete subset of) $[0,2\pi)^d$, with $d=2^n$,.
The corresponding 2-fold twirl maps are $\cG_Z$ and $\cG_X$.
They define a compound map $\cR=\cG_Z\circ \cG_X \circ \cG_Z$,
which represents a 2-fold twirl with unitaries of the form $U_Z(\qf')U_X(\chi)U_Z(\qf)$,
where $\qf',\chi,\qf$ are sampled independently.
Applying the map $\cR$ multiple times yields a more complicated 2-fold twirl.
The map $\cR^\ell$ has unitaries of the form
$U_Z(\qf_\ell')U_X(\chi_\ell)U_Z(\qf_\ell)$ $\cdots$ $U_Z(\qf_1')U_X(\chi_1)U_Z(\qf_1)$.
Nakata \emph{et al.}\,\cite{Nakata2017Unitary2Designs} showed that their construction
is an $\qe$-approximate 2-design with $\qe$ proportional to $d^{-\ell}$.
Roughly speaking one can think of $\cG_Z$ and $\cG_X$ as non-commuting projections.
Projecting a vector back and forth a couple of times reduces the length of the vector.
In the same way, acting on a matrix alternatingly with $\cG_Z$ and $\cG_X$
reduces components of the matrix.

In the study of $t$-designs one is led to the representation theory of the symmetric group $S_t$ (permutation of $t$ elements).
Let the vectors $\ket a$ form an orthonormal basis of $d$-dimensional Hilbert space~$\cH_d$, and let
$\ket{ab}=\ket a\otimes\ket b$. 
The following operators on $\cH_d\otimes\cH_d$ are a representation of $S_2$ and play a central role in the construction of
unitary two-designs,
\be
    I=\sum_{ab}\ket{ab}\bra{ab}\qquad\qquad 
    F = \sum_{ab}\ket{ab}\bra{ba}.
\ee
$I$ is the $2d\times 2d$ identity matrix, and $F$ is the swap operator.
Note that $F^2=I$.

We denote the span of $I$ and $F$ as $\cA$, and we define $\cK=\cA^\perp$.
For all $Y\in\cK$ it holds that $\langle Y,I\rangle=0$ and $\langle Y,F\rangle=0$.
For all $X\in\cL(\cH_d\otimes\cH_d)$ there exist unique $X_\cA\in\cA$, $X_\cK\in\cK$ such that
$X=X_\cA+X_\cK$ and $\langle X_\cA,X_\cK \rangle=0$.

From the definition of a two-fold twirl (\ref{eq:twirl-def}) it follows that
$\cG_\mu^{(2)} (I)=I$ and $\cG_\mu^{(2)} (F)=F$ for any~$\mu$.
In other words, the subspace $\cA$, spanned by $\{I,F\}$, is invariant under any two-fold twirl.
Furthermore, it is also clear from (\ref{eq:twirl-def}) that
$\Tr I\cG_\mu^{(2)}(X) = \Tr\, IX$ and $\Tr\, F\cG_\mu^{(2)}(X) = \Tr\, FX$.
Hence, $\Tr\, IX=0 \wedge \Tr\, FX=0$ implies $\Tr\, I\cG_\mu^{(2)}(X) =0 \wedge \Tr\, F\cG_\mu^{(2)}(X) =0$.
In other words, the subspace $\cK$ is closed under any $\cG_\mu^{(2)}$.
This yields
\be
    \forall_\mu \quad
    \cG_\mu^{(2)}(X_\cA+X_\cK) = X_\cA + \cG_\mu^{(2)}(X_\cK)
    \quad \wedge\quad \cG_\mu^{(2)}(X_\cK) \in\cK.
\ee
The term $X_\cA$ is what the Haar measure would produce.
Hence, in the construction of an (approximate) unitary two-design
one has to consider only how the mapping acts on the $\cK$ subspace.
\begin{corollary}
\label{corol:Kspace}
Let $\cG_{\mu,\cK}^{(2)}:\cK\to\cK$ denote the mapping $\cG_\mu^{(2)}$ restricted to the space~$\cK$.
If $\| \cG_{\mu,\cK}^{(2)} \|_{2\to2}\leq \qe$ then $\mu$ is an $\qe$-approximate unitary two-design.
\end{corollary}

%-----------------------------------------------------------------
\subsection{Continuous Variables}
\label{subsec:prelim-CV}

A {\em mode} of the electromagnetic vector potential
represents a plane wave solution of the vacuum Maxwell equations
at a certain frequency, wave vector and polarisation.
Associated with each mode there is a creation operator $\hat a^\dagger$
and annihilation operator~$\hat a$.
The linear combinations $\hat q=\frac{\hat a+\hat a^\dagger}{\sqrt2}$
and $\hat p=\frac{\hat a-\hat a^\dagger}{i\sqrt 2}$
are easy-to-observe quantities called {\em quadratures}, and they behave
as the position and momentum operator of a harmonic oscillator with Hamiltonian
$\frac{\hat p^2+\hat q^2}2$ (assuming $\hbar=1$).
The operators $\hat p$ and $\hat q$ do not commute.
For $x\in\RR$ let $\ket x$ be the eigenstate of $\hat q$ satisfying $\hat q\ket x=x\ket x$.
The set of states $\{\ket x\}_{x\in\RR}$ forms a basis of the Hilbert space.
The Hilbert space is infinite-dimensional.
The q-basis is pseudo-normalised, satisfying $\braket xy=\qd(x-y)$, where $\qd$ is a Dirac delta function
(instead of the usual Kronecker delta).
Similarly, the eigenstates of $\hat p$ form a pseudo-normalised basis $\{\ket{\underline k}\}_{k\in\RR}$.
These two bases are each other's Fourier transform, $\braket{x}{\underline k}=e^{ikx}/\sqrt{2\pi}$.

\section{An approximate unitary two-design for Continuous Variables}
\label{sec:ourdesign}

%------------------------------------------------------------------
\subsection{Discretisation of the CV space}
\label{sec:discret}

A one-mode CV density matrix can be written as $\hat\qr=\int\rd x\rd x'\, \ket x\bra{x'} \qr(x,x')$,
where the states $\ket x$ are the eigenstates of the q-quadrature operator.
We will consider density matrices $\qr$ that satisfy two conditions:
(i) the derivatives $\frac{\partial}{\partial x}\qr(x,x')$ and $\frac{\partial}{\partial x'}\qr(x,x')$
are bounded in absolute value;
(ii) there exists a value $q_{\rm max}$ such that $\int_{|x|+|x'|>x_{\rm max}}\rd x  \rd x' \; \qr(x,x') \ll 1$.

The first condition allows us to approximate $\qr$ as piecewise constant on small intervals of size~$\qD$.
The second condition allows us to work with a finite number ($d$) of intervals, stopping at a distance  $\approx q_{\rm max}$ away from the origin.
We take $d$ to be even, without loss of generality.
We define discrete q-eigenstates denoted as $\ket{i}$, with $i\in \cI_d$,
$\cI_d=\{-\frac{d}{2},\ldots, \frac{d}{2}-1 \}$, as
$\ket i = \frac1{\sqrt\qD}\int_{0}^{\qD}\rd x\,\ket{q=i\qD+x}$.
We refer to these as `box kets'.
It holds that $\braket{i}{j}=\qd_{ij}$.
This is verified as follows.
$\braket ij=\frac1\qD \int_{0}^{\qD}\rd x\int_{0}^{\qD}\rd y\,\braket{q=i\qD+x}{q=j\qD+y}$
$=\frac1\qD \int_{0}^{\qD}\rd x\int_{0}^{\qD}\rd y\, \qd_{ij}\qd(x-y)$
$=\qd_{ij}\frac1\qD \int_{0}^{\qD}\rd x=\qd_{ij}$.

We set $\qD=\sqrt{2\pi/d}$, which gives $q_{\rm max}=\sqrt{\pi d/2}$.
The states $\{\ket i\}_{i\in\cI_d}$ span a $d$-dimensional Hilbert space~$\cH_d$.
With this discretisation, the density matrix is approximated as 
$\hat\qr = \sum_{i,i'\in\cI_d} \qr^{i'}_i \ket i\bra{i'}$.

We construct p-quadrature eigenstates $\widetilde{\ket\cdots}$ as the discrete Fourier transform of the $q$-eigenstates.
Let $\qo=\exp (i\frac {2\pi}d)$ be a $d$'th root of unity.
For $k\in\cI_d$ we define
$\widetilde{\ket{ k}} \stackrel{\rm def}{=} \frac1{\sqrt d}\sum_{j\in\cI_d} \qo^{kj}\ket j$.
It holds that $\widetilde{\bra{k'}}\widetilde{\ket k}=\qd_{k'k}$.

All operators on the original CV Hilbert space are mapped to $\cH_d$
by applying the projection operator $\Pi=\sum_{i\in\cI_d} \ket i \bra i$
$=\sum_{k\in\cI_d} \widetilde{\ket k}\widetilde{\bra k}$.
We introduce scaled quadrature operators that produce the integer label of an interval,
\be
    \hat Q \stackrel{\rm def}{=} \sum_{i\in\cI_d} i \ket i \bra i
    \quad\quad
    \hat P \stackrel{\rm def}{=} \sum_{k\in\cI_d} k \widetilde{\ket k}\widetilde{\bra k}.
\ee
The CV `position' operator is $\hat q=\int_{-\infty}^\infty\!\rd q\; q\ket q\bra q$.
Projection of $\hat q$ gives
\bea
    \Pi \hat q \Pi^\dagger &=& \sum_{j\in\cI_d}\int_0^\qD\!\rd x\; (j\qD+x) \Pi \ket{q=j\qD+x}\bra{q=j\qD+x} \Pi
    \\
    &=& \sum_{j\in\cI_d}\int_0^\qD\!\rd x\; (j\qD+x) \frac{\ket j\bra j}{\qD}
    = \qD \sum_{j\in\cI_d} (j+\frac12) \ket j\bra j
    = \hat Q \qD +\frac \qD2\Pi.
\eea
We show that the error made in the transition from a state $\qr$ to its discretised version is
related to the expected energy of the state divided by the number of discretisation intervals.

\begin{lemma}
\label{lemma:energy}
Let $\qr$ be a density matrix with matrix elements $\qr(q,q')=\bra q \qr \ket{q'}$. 
Let $\qr(q,q')$ be negligible outside of the interval 
$|q|< \sqrt{\pi d/2} \wedge |q'|< \sqrt{\pi d/2}$. 
.
Let $\qr_{\rm discr}=\frac{\Pi \qr \Pi^\dagger}{\Tr (\Pi \qr \Pi^\dagger)}$ be the discretised version of~$\qr$.
The trace distance between the discretised and the original state is upper bounded as
\be
    \| \qr_{\rm discr}-\qr \|_{\rm tr}  \leq  
    \sqrt{\frac2\pi\cdot\frac{\Tr\; \hat p^2 \, \qr}{d}}.
\ee
where $\hat p$ is the CV p-quadrature operator.
\end{lemma}

\begin{proof}
We first consider a pure state $\ket\qj=\int\rd q\, \ket q \qj(q)$. 
Let $x\in[0,\qD)$.
We note that $\braket{i}{q=i\qD+x}=\frac1{\sqrt\qD}$
and $\Pi \ket{q=i\qD+x} = \frac1{\sqrt\qD}\ket i$.
Projection of $\ket\qj$ gives
$
    \Pi \ket\qj = \sum_j \int_0^\qD \! \rd x\; \frac{\ket j}{\sqrt\qD}\qj(j\qD+x)
    = \sum_j \ket j \, \bar\qj_j\sqrt\qD,
$
where we have defined $\bar\qj_j=\int_0^\qD \rd x\; \qj(j\qD+x)$.
This yields
\be
    \bra\qj \Pi^\dagger \Pi \ket\qj=\sum_j |\bar\qj_j|^2 \qD.
\label{psiproj}
\ee
Clearly, $|\bar\qj_j|^2 \qD$ is the probability mass in the $j$'th box.
We write
\bea
    1-\Tr\, \Pi \ket\qj\bra\qj \Pi^\dagger &=&   
    1-\bra\qj\Pi^\dagger \Pi \ket\qj
    \\ &\stackrel{(\ref{psiproj})}=&
    \sum_j \int_0^\qD\rd x\, |\qj(j\qD+x)|^2 - \sum_j  \int_0^\qD\rd x\, |\bar\qj_j|^2
    \\ &=&
    \sum_j \int_0^\qD\rd x\, | \qj(j\qD+x) - \bar\qj_j |^2 
    \\ & \stackrel{\rm (a)}{\leq} &
    \sum_j \frac{\qD^2}{\pi^2}\int_0^\qD\rd x\, |\qj'(j\qD+x)|^2 
    = \frac{\qD^2}{\pi^2} \int\!\rd q \; \left|\frac{\rd\qj}{\rd q}\right|^2
    \\ &\stackrel{\rm (b)}=&
    \frac{\qD^2}{\pi^2} \bra\qj \hat p^2 \ket\qj.
\label{epspure}
\eea
In (a) we used the Poincar\'e--Wirtinger inequality \cite{HLP1952,poincare}
for a one-dimensional interval of size~$\qD$.
In (b) we used $\hat p=\frac1i \frac{\rd}{\rd q}$ and integration by parts,
with $\qj(q)\qj'(q)=0$ at the integration boundaries.
The result (\ref{epspure}) generalizes straightforwardly from pure states $\ket\qj\bra\qj$ 
to mixed states $\sum_r \ql_r \ket{\qj_r}\bra{\qj_r}$,
yielding
$1-\Tr\, \Pi \qr \Pi^\dagger\leq \frac{\qD^2}{\pi^2}\sum_r \ql_r \bra{\qj_r} \hat p^2 \ket{\qj_r}$.
Finally we use the Gentle Measurement Lemma (Lemma~\ref{lem:gentle-m})
and $\qD^2=2\pi/d$.
\end{proof}
Since $\frac12 \hat p^2+\frac12\hat q^2=\hat n+\frac12$, with $\hat n$ the photon number operator,
we have $\Tr \qr \hat p^2 \leq 2\bar n+1$, where $\bar n$ is the expected photon number.
This establishes a relation between the error caused by the discretisation and the energy of the state, 
leading to a discretisation error less than $\sqrt{\frac4\pi(\bar n+\frac12)/d}$.
Hence, $d$ should be chosen large enough to accommodate~$\bar n$.
Conversely, given a fixed $d$ this relation indicates which states are going to be 
approximated well.

%--------------------------------------------------------------------------------
\subsection{Approximate unitary 2-design compatible with quadratures}
\label{sec:construction}

We construct an approximate two-design on~$\cH_d$ that is based on the $\hat P$ and $\hat Q$ quadrature operators.
The space $\cH_d$ represents the CV space according to the discretisation procedure
given in Section~\ref{sec:discret}.
Alternatively, one may forget about the connection with the CV Hilbert space  
and just consider the finite-dimensional space $\cH_d$ as a thing on its own.

The main ingredient in the design
is operators of the form $V_{\qa\qb}=\qo^{\qa\hat Q+\qb\hat Q^2}$
and $\tilde V_{\qa\qb}=\qo^{\qa\hat P+\qb\hat P^2}$, with continuous variables $\qa,\qb\in[0,d)$ drawn uniformly at random.\footnote{
The interval can also be chosen as e.g. $[-\frac d2, \frac d2)$.
}
Let $X$ be an operator on $\cH_d\otimes \cH_d$.
We write $X=\sum_{iji'j'} X^{i'j'}_{ij} \ket{ij}\bra{i'j'}$, where $\ket{ij}$ is shorthand for $\ket i\otimes \ket j$.
We define superoperators $G$ and $\tilde G$ as
\begin{eqnarray}
    G(X) &=& \EE_{\qa\qb}\; V_{\qa\qb}^{\otimes 2}   \; X \; (V_{\qa\qb}^\dagger)^{\otimes 2}
    \\
    \tilde G(X) &=& \EE_{\qa\qb}\; \tilde V_{\qa\qb}^{\otimes 2}   \; X \; (\tilde V_{\qa\qb}^\dagger)^{\otimes 2}.
\end{eqnarray}
The alternating application of these superoperators is denoted as
\begin{equation}
    \cR(X) = (G\circ \tilde G \circ G)(X) 
    = \EE_{\qa\qb\qa'\qb'\qa''\qb''} \; (V_{\qa''\qb''} \tilde V_{\qa'\qb'} V_{\qa\qb})^{\otimes 2} \; X\; 
    (V_{\qa\qb}^\dagger \tilde V_{\qa'\qb'}^\dagger V_{\qa''\qb''}^\dagger)^{\otimes 2}.
\end{equation}
The construction is inspired by \cite{Nakata2017Unitary2Designs}, but with a number of  differences: 
(i)
Our random unitaries $V_{\qa\qb},\tilde V_{\qa\qb}$ have two parameters, whereas in \cite{Nakata2017Unitary2Designs}
the unitaries are of the form $U_Z(\vec\qf)=\sum_{j=0}^{2^n-1} e^{i\qf_j}\ket j\bra j$ 
(diagonal in the computational basis)
and $U_X(\vec\chi)=H^{\otimes n}\sum_{j=0}^{2^n-1} e^{i\chi_j}\ket j\bra j H^{\otimes n}$
(diagonal in the Hadamard basis)
where the $\qf_j, \chi_j$ are independently random for all~$j$.  
(ii)
The computational basis in \cite{Nakata2017Unitary2Designs} may be interpreted as the q-quadratrure in a CV system.
However, there is no clear CV interpretation of the Hadamard basis.
In our construction both $V_{\qa\qb}$ and $\tilde V_{\qa\qb}$ are related to quadrature operators.

We will now prove that $\cR$ acts almost like a two-design,
and that the resemblance becomes better with increasing~$d$.
Furthermore, applying $\cR$ again gives an improvement by a factor $\frac1d$. 
Let us introduce the notation 
\be
    E = \sum_i \ket{ii}\bra{ii}.
\ee

\begin{lemma}
\label{lemma:basicaction}
Let $i\in\cI_d$ and $u\in\cI_d\setminus\{0\}$.
For $\{i',j'\}\neq \{i,j\}$ it holds that $\cR(\ket{ij}\bra{i'j'})=0$. 
Furthermore,
\begin{eqnarray}
    \cR(\ket{i,i+u}\bra{i,i+u}) &=& \frac1{d^2} \sum_a\ket{a,a+u}\bra{a,a+u}+\frac I{d^2} - \frac{I+F-E}{d^3}
\label{RcaseI}
    \\
    \cR(\ket{i,i+u}\bra{i+u,i}) &=& \frac1{d^2} \sum_a\ket{a,a+u}\bra{a+u,a}+\frac F{d^2} - \frac{I+F-E}{d^3}
    \\
    \cR(\ket{ii}\bra{ii}) &=& \frac{I+F}{d^2}  - \frac{I+F-E}{d^3},
\label{RcaseIII}    
\end{eqnarray}
where the additions $i+u$, $a+u$ are understood to be modulo $d$, resulting in values in $\cI_d$. 
\end{lemma}
\begin{proof}
$G(\ket{ij}\bra{i'j'}) 
= \ket{ij}\bra{i'j'} \EE_{\qa\qb}\qo^{\qa(i+j-i'-j')}\qo^{\qb(i^2+j^2-i'^2-j'^2)} $
$= \ket{ij}\bra{i'j'} \frac1d \int_0^d\rd\qa\, e^{i2\pi\frac\qa d(i+j-i'-j')}$ 
\newline
$\frac1d\int_0^d\rd\qb\, e^{i2\pi \frac\qb d(i^2+j^2-i'^2-j'^2)}$
$=\ket{ij}\bra{i'j'} \qd_{i'+j',i+j}\qd_{i'^2+j'^2,i^2+j^2}$.
The Kronecker deltas can equal 1 only if $\{i',j'\} = \{i,j\}$.
There are three possibilities:
(I) $j\neq i$ and $(i',j')=(i,j)$;
(II) $j\neq i$ and $(i',j')=(j,i)$;
(III) $j=i$ and ${i',j'}=(i,i)$.
These correspond to the lines (\ref{RcaseI}) to (\ref{RcaseIII}).
We introduce the notation
$G(\ket{ij}\bra{i'j'}) =\ket{ij}\bra{i'j'} C_{ij}^{i'j'} $.
Next we do the reverse Fourier transform to write
$\ket{ij}\bra{i'j'} =$ $\frac1{d^2}\sum_{pkp'k'} \widetilde{\ket{pk}}\widetilde{\bra{p'k'}} \qo^{-ip-jk+i'p'+j'k'}$.
We act with $\tilde G$ on $G(X)$. 
In analogy with the q-quadrature case, it holds that
$\tilde G (\widetilde{\ket{pk}}\widetilde{\bra{p'k'}})$ $=\widetilde{\ket{pk}}\widetilde{\bra{p'k'}} C_{pk}^{p'k'}$.
This gives
\bea
    (\tilde G\circ G)(\ket{ij}\bra{i'j'}) &=& 
    C_{ij}^{i'j'}\frac1{d^2}\sum_{pkp'k'} \widetilde{\ket{pk}} \widetilde{\bra{p'k'}} C_{pk}^{p'k'} \qo^{-pi-kj+p'i'+k'j'}
    \\ &=&
    C_{ij}^{i'j'} \frac1{d^4}\sum_{aba'b'}\ket{ab}\bra{a'b'}\sum_{pkp'k'}C_{pk}^{p'k'}
    \qo^{p(a-i)+k(b-j)+p'(i'-a')+k'(j'-b')},
\eea
where in the last step we did the Fourier transform.
Applying $G$ again inserts a factor $C_{ab}^{a'b'}$ into the $ab$-summation.
The `free' summation over $pkp'k'$ yields three cases that each give rise to a particular combination of Kronecker deltas,
$\frac1{d^2}\sum_{pkp'k'}C_{pk}^{p'k'} \qo^{p(a-i)+k(b-j)+p'(i'-a')+k'(j'-b')}$
$=\qd^{({\rm mod }\,d)}_{a'-a,i'-i}\qd^{({\rm mod }\,d)}_{b'-b,j'-j} 
+ \qd^{({\rm mod }\,d)}_{a-b',i-j'}\qd^{({\rm mod }\,d)}_{b-a',j-i'} - \frac1d$.
The ``mod $d$'' occurs because all values in the exponent of $\qo$ are integers.  
The $\frac1d$ term originates from case (III) where the delta automatically evaluates to 1
because of the presence of the constraint $C_{ab}^{a'b'}$.
The end result follows by handling the three cases of $C_{ij}^{i'j'}$ separately.
\end{proof}

In (\ref{RcaseI})-(\ref{RcaseIII}) we note that the right-hand sides do not depend on~$i$.
This makes it particularly easy to compute the action of $\cR$ on some general matrix~$X$.

\begin{corollary}
Let $X\in \cL(\cH_d^{\otimes 2})$. It holds that
\bea
    \cR(X) &=&
    \frac1{d^2}\sum_{u\neq 0} (\sum_i X_{i,i+u}^{i,i+u}) \sum_a  \ket{a,a+u}\bra{a,a+u} 
    +\frac1{d^2}\sum_{u\neq 0} (\sum_i X_{i,i+u}^{i+u,i}) \sum_a  \ket{a,a+u}\bra{a+u,a} 
    \nonumber\\ &&  
    + (\sum_i X_{ii}^{ii})\frac{I+F-E}{d^3}
    +(\frac I{d^2}-\frac{I+F-E}{d^3})\Tr X
    +(\frac F{d^2}-\frac{I+F-E}{d^3})\Tr FX.
\eea
For $X\in \cK$ we can write
\bea
    X\in \cK\!: \quad\quad \cR(X) &=& 
     \sum_{u\neq 0} \Big(\frac1{d^2}\sum_i X_{i,i+u}^{i,i+u} + \frac{\sum_i X_{ii}^{ii}}{d^3}\Big)\sum_a\ket{a,a+u}\bra{a,a+u} 
    \nonumber\\ &&
    +  \sum_{u\neq 0} \Big(\frac1{d^2}\sum_i X_{i,i+u}^{i+u,i} + \frac{\sum_i X_{ii}^{ii}}{d^3}\Big)\sum_a\ket{a,a+u}\bra{a+u,a} 
    +\frac{\sum_i X_{ii}^{ii}}{d^3}E.
\label{componentform}
\eea
\end{corollary}
\begin{proof}
Follows directly from $\cR(X)=\sum_{iji'j'}X_{ij}^{i'j'}\cR(\ket{ij}\bra{i'j'})$
and using Lemma~\ref{lemma:basicaction}.
\end{proof}

The form (\ref{componentform}) allows us to directly compute the norm
$\| \cR(X) \|_2$.

\begin{theorem}
\label{th:latticeDesign}
The set of unitaries 
$\{\qo^{\qa''\hat Q+\qb''\hat Q^2} \qo^{\qa'\hat P+\qb'\hat P^2} \qo^{\qa\hat Q+\qb\hat Q^2}\}$,
%$\{V_{\qa''\qb''} \tilde V_{\qa'\qb'} V_{\qa\qb} \}$
with $\qa,\qb,\qa',\qb',\qa'',\qb''$ chosen uniformly in $[0,d)$,
is an $\qe$-approximate unitary two-design with $\qe=\frac1d$.
\end{theorem}

\begin{proof}
Consider $X\in \cK$.
We introduce notation $Y=\cR(X)$ and  $c=\sum_i X^{ii}_{ii}$.
From $\Tr\, Y^\dagger Y=\sum_{iji'j'}|Y_{ij}^{i'j'}|^2$ and
(\ref{componentform}) we obtain
\begin{eqnarray}
    \Tr\, Y^\dagger Y &=&
    \frac1{d^4}\sum_a \sum_{u\neq 0}\Big|\sum_i X_{i,i+u}^{i,i+u} +\frac cd\Big|^2
    + \frac1{d^4}\sum_a \sum_{u\neq 0}\Big|\sum_i X_{i,i+u}^{i+u,i} +\frac cd\Big|^2 + d \frac{|c^2|}{d^6}
    \\ &=&
    \frac1{d^3} \sum_{u\neq 0}\Big|\sum_i X_{i,i+u}^{i,i+u}\Big|^2
    +\frac1{d^3} \sum_{u\neq 0}\Big|\sum_i X_{i,i+u}^{i+u,i}\Big|^2
    +2d(d-1) \frac{|c|^2}{d^6}
    \nonumber\\ &&
    + (\frac{c^*}{d^4}\sum_{u\neq 0} \sum_i X_{i,i+u}^{i,i+u} + cc)
    + (\frac{c^*}{d^4}\sum_{u\neq 0} \sum_i X_{i,i+u}^{i+u,i} + cc)
    + d \frac{|c^2|}{d^6}
\label{RXsquare1}    
    \\ &\stackrel{\rm (a)}{=}& 
    -\frac{|c|^2}{d^5}-2\frac{|c|^2}{d^4}
    +\frac1{d^3}\sum_{u\neq 0} \left| \sum_s X^{s,s+u}_{s,s+u}\right|^2
    + \frac1{d^3}\sum_{u\neq 0} \left| \sum_s X^{s+u,s}_{s,s+u}\right|^2
    \\ & \stackrel{\rm (b)}{\leq} &
    \frac1{d^3}\sum_{u\neq 0} (\sum_s |X^{s,s+u}_{s,s+u}|  )^2
    +\frac1{d^3}\sum_{u\neq 0} (\sum_s |X^{s+u,s}_{s,s+u}| )^2
    \\ & \stackrel{\rm (c)}{\leq} &
    \frac1{d^2}\sum_{u\neq 0}\sum_s |X^{s,s+u}_{s,s+u}|^2
    + \frac1{d^2}\sum_{u\neq 0}\sum_s |X^{s+u,s}_{s,s+u}|^2
    \\ &\stackrel{\rm (d)}{\leq}&
    \frac1{d^2}\Tr X^\dagger X.
\end{eqnarray}
The $cc$ stands for complex conjugate.
In (a) we used $\sum_{u\neq 0}\sum_i X_{i,i+u}^{i,i+u} = \Tr X -c = -c$
and $\sum_{u\neq 0}\sum_i X_{i,i+u}^{i+u,i} = \Tr FX -c = -c$.
Step (b) is the triangle inequality, and neglecting the negative terms.
In (c) we have used $\|x\|_1^2 \leq d \| x\|_2^2$.
Finally, (d) follows from $\Tr X^\dagger X=\sum_{aba'b'}|X^{a'b'}_{ab}|^2$.
We conclude that $\frac{\Tr \cR^\dagger(X) \cR(X)}{\Tr X^\dagger X} \leq \frac1{d^2}$
and hence $\sup_{X\in\cK, X\neq 0} \frac{\| \cR(X)\|_2}{\| X \|_2} \leq \frac1d$.
Finally we use Corollary~\ref{corol:Kspace}.
\end{proof}

\begin{lemma}
\label{lemma:recurseR}
Let $X\in \cK$. 
For $\ell \geq 1$, applying $\cR$ $\ell$ times results in 
\bea
    \cR^\ell(X) & = &
    \sum_{u\neq0} \xi_u^{(\ell)}\sum_a \ket{a,a+u}\bra{a,a+u}
    + \sum_{u\neq0} \ql_u^{(\ell)}\sum_a \ket{a,a+u}\bra{a+u,a} + \qh^{(\ell)} \sum_a \ket{aa}\bra{aa}
\label{RXellform}
    \\ 
\label{xiuell}
    \xi_u^{(\ell)} &=& \frac1{d^{\ell+1}}\sum_i X_{i,i+u}^{i,i+u} 
    +\frac{\sum_i X_{ii}^{ii}}{d^{\ell+2}} \sum_{t=0}^{\ell-1}(\frac1d)^t 
    = \frac{\sum_i X_{i,i+u}^{i,i+u} }{d^{\ell+1}} + 
    \frac{\sum_i X_{ii}^{ii}}{d^{\ell+1}}\cdot\frac{1-(1/d)^{\ell}}{d-1} 
    \\
    \ql_u^{(\ell)} &=& \frac1{d^{\ell+1}}\sum_i X_{i,i+u}^{i+u,i} + \frac{\sum_i X_{ii}^{ii}}{d^{\ell+2}}\sum_{t=0}^{\ell-1}(\frac1d)^t
    = \frac{\sum_i X_{i,i+u}^{i+u,i} }{d^{\ell+1}} + 
    \frac{\sum_i X_{ii}^{ii}}{d^{\ell+1}}\cdot\frac{1-(1/d)^{\ell}}{d-1}
\label{lambdauell}
    \\
    \qh^{(\ell)} &=& \frac{\sum_i X_{ii}^{ii}}{d^{2\ell+1}}.
\label{etaell}
\eea
\end{lemma}

\begin{proof}
Eq.(\ref{componentform}) has the same structure as (\ref{RXellform}),
with $\xi_u^{(1)}=\frac1{d^2}\sum_i X_{i,i+u}^{i,i+u} + \frac1{d^3}\sum_x X_{ii}^{ii}$,
$\ql_u^{(1)}=\frac1{d^2}\sum_i X_{i,i+u}^{i+u,i} + \frac1{d^3}\sum_x X_{ii}^{ii}$,
$\qh^{(1)}=\frac1{d^3}\sum_x X_{ii}^{ii}$.
Eq.(\ref{componentform}) applied to itself defines a recursion relation
$\qh^{(s+1)}=\qh^{(s)}/d^2$, 
$\xi_u^{(s+1)}=\frac1d \xi_u^{(s)}+\frac1{d^2}\qh^{(s)}$,
$\ql_u^{(s+1)}=\frac1d \ql_u^{(s)}+\frac1{d^2}\qh^{(s)}$.
Solving the recursion yields the result of the lemma.
\end{proof}

\begin{corollary}
Applying $\cR^\ell$ is an $\qe$-approximate unitary two-design with $\qe=(\frac1d)^\ell$.
\end{corollary}
\begin{proof}
Let $\cK' = \cK\setminus\{0\}$.
Let $\cR_\cK$ denote the restriction of $\cR$ to the $\cK$ subspace.
We apply $\| \cR_\cK \|_{2\to 2}\leq\frac1d$
recursively:
$ \|  \cR^\ell_\cK\|_{2\to 2} \stackrel{\rm def}{=}$ $ \sup_{X\in\cK'} \frac{\| \cR^\ell(X) \|_2}{\| X\|_2} $
$= \sup_{X\in\cK'} \frac{\| \cR(\cR^{\ell-1}(X)) \|_2}{ \|\cR^{\ell-1}(X)\|_2 } \frac{ \|\cR^{\ell-1}(X)\|_2 }{\| X\|_2}$
$\leq (\sup_{Y\in\cK'} \frac{\| \cR(\cR^{\ell-1}(Y)) \|_2}{ \|\cR^{\ell-1}(Y)\|_2 }) $
$ (\sup_{X\in\cK'} \frac{ \|\cR^{\ell-1}(X)\|_2 }{\| X\|_2})$
$\leq \| \cR_\cK \|_{2\to 2} \cdot \| \cR^{\ell-1}_\cK \|_{2\to 2}$
$\leq \frac1d \| \cR^{\ell-1}_\cK \|_{2\to 2}$ etc.
This establishes that $\|  \cR^\ell_\cK\|_{2\to 2}\leq 1/d^\ell$.
Finally we use Corollary~\ref{corol:Kspace}.
\end{proof}

%---------------------------------------------------------------------------------------

%---------------------------------------------------------------------------------------
\subsection{Alternative construction with discrete $\boldsymbol{\qa,\qb}$}
\label{sec:prime}

It is possible to modify the scheme so that the parameters $\qa,\qb$ are chosen from a finite set.
The obvious advantage is that the set of unitaries becomes finite.
The disadvantage is that we are then forced to work with an odd dimension $d$, which does not fit nicely 
with the Unclonable Encryption construction of \cite{BC2026}.
This is further discussed in Section~\ref{sec:UE}.

In this variant, we take $d$ to be an odd prime.
We redefine $\cI_d$ to $\cI^{\rm odd}_d=\{-\frac{d-1}2,\ldots,\frac{d-1}2\}$.
The $\qa,\qb$ are now sampled uniformly from $\cI^{\rm odd}_d$.
The q-intervals in CV space are redefined as $(i\qD-\frac\qD2, i\qD+\frac\qD2)$.
{\em All the results presented in Section~\ref{sec:construction}
still hold, with identical proofs.}
It is important to note that the identity
$\EE_{\qa\qb}\qo^{\qa(i+j-i'-j')}\qo^{\qb(i^2+j^2-i'^2-j'^2)} $
$= \qd_{i'i}\qd_{j'j}+\qd_{i'j}\qd_{j'i}-\qd_{ji}\qd_{i'i}\qd_{j'i}$,
which plays a crucial role in the scheme,
can be made to hold for discrete $\qa,\qb$ only when $d$ is prime.
For discrete $\qa,\qb$, summation yields
$\frac1{d^2}\sum_{\qa,\qb\in\cI^{\rm odd}_d} \qo^{\qa(i+j-i'-j')}\qo^{\qb(i^2+j^2-i'^2-j'^2)} $
$= \qd^{({\rm mod}\, d)}_{i'+j',i+j}\qd^{({\rm mod}\, d)}_{i'^2+j'^2,i^2+j^2}$
where the Kronecker deltas are defined modulo~$d$.
If $d$ were factorisable as $d=pq$, additional solutions of the form
$(i',j')=(i+p,j-p)$ would be valid in case $j-i=p+q\cdot{\rm integer}$.
In order to avoid this complication, we choose $d$ to be prime.

%-----------------------------------------------------------------

%========================================================================
\section{Physical implementation of boxed phase unitaries in photonic CV}
\label{sec:photonicCV}

In the photonic CV setting, the boxed unitaries
$V_{\qa\qb},\tilde V_{\qa\qb}$
should be understood as \emph{phase profiles that are constant on each quadrature bin}.
This perspective is useful because it identifies $V_{\alpha\beta}$ and $\tilde V_{\alpha\beta}$
as instances of (non-Gaussian) \emph{quadrature-phase gates} $e^{i f(\hat q)}$ and
$e^{i f(\hat p)}$, for explicitly defined staircase functions~$f$.
We can write
\be
    V_{\qa\qb} = e^{if_{\qa\qb}(\hat q)}, \quad\quad
    f_{\qa\qb}(q) = 2\pi\Big(\frac\alpha d\,\lfloor q/\Delta\rfloor+\frac \beta d\,\lfloor q/\Delta\rfloor^2\Big)
    \mbox{ for } 
    q\in[-\sqrt{\pi d/2}, \sqrt{\pi d/2}). 
\ee
It should be noted that $\exp i2\pi \frac{\qb}{d}\lfloor \frac q\qD\rfloor^2$ is typically {\em not}
close to the Gaussian expression $\exp i2\pi \frac{\qb}{d}( \frac q\qD)^2$;
the latter expression can have many oscillations in a q-interval of width~$\qD$, whereas the former stays constant.
Hence our scheme {\it cannot} be implemented with Gaussian operators in $\hat p$ and~$\hat q$.

In traveling-wave photonic CV, implementing $e^{i f(\hat q)}$ for a general real function $f$
is a ``nonlinear phase gate'' task. 
The prevailing implementation is \emph{measurement-induced} (or teleportation-based): 
one combines linear optics with ancilla states,
quadrature measurements, and classical feedforward so that the desired non-Gaussian transformation
is applied deterministically (or near-deterministically) to the target mode \cite{Menicucci2006}.
This approach is rooted in continuous-variable teleportation \cite{BraunsteinKimble1998}, which is been experimentally well established \cite{Furusawa1998},
and measurement-based CV processing \cite{BraunsteinVanLoock2005}.

Fortunately,
\emph{one does not need to implement a discontinuous staircase gate perfectly}; 
it suffices to implement an approximation on the finite discretisation window.
There are two complementary theoretical guarantees that justify this.

\smallskip
\noindent
(i) \emph{Universality from polynomial Hamiltonians.}
Lloyd and Braunstein \cite{LloydBraunstein1999}
have shown that Gaussian operations together with an appropriate nonlinearity
(e.g.\ a cubic phase interaction, Kerr-type nonlinearity, etc.) yield universality for
transformations generated by polynomial Hamiltonians in $\hat q,\hat p$; in particular one can
approximate single-mode unitaries of the form $e^{i {\rm poly}(\hat q)}$  to
high accuracy  \cite{LloydBraunstein1999,BraunsteinVanLoock2005},
where even high-order polynomials can be implemented in a limited number of steps.
Since any bounded function on a compact interval can be uniformly approximated by polynomials,
one can approximate the staircase profile $f_{\alpha\beta}(q)$ on a finite window by a polynomial
$z_{\alpha\beta}(q)$, and then target $e^{i z_{\alpha\beta}(\hat q)}$
as the implementable proxy\footnote{
And similarly for $\hat p$. What can be achieved for one quadrature can also be achieved for the other.
} 
for $e^{i f_{\alpha\beta}(\hat q)}$.

\smallskip
\noindent
(ii) \emph{Direct measurement-induced synthesis of nonlinear quadrature phase gates.}
Beyond universality as an existence statement, there are explicit measurement-induced protocols
to realize nonlinear quadrature phase gates.
Filip et al. \cite{FilipMarekAndersen2005} have developed a measurement-induced toolbox for CV interactions using linear optics,
homodyne detection, and offline resources.
Marek and Filip \cite{MarekFilip2011} have proposed a deterministic method for implementing weak cubic nonlinearity using ancilla
states and feedforward.
Miyata \emph{et al.} propose an implementation of a cubic gate by adaptive non-Gaussian measurement \cite{Miyata2016}.
Most directly aligned with our needs, Marek \emph{et al.}\;\cite{Marek2018} presented a \emph{general methodology for deterministic
realization of arbitrary-order nonlinear quadrature phase gates}, enabling direct implementations of
single-mode transformations that nonlinearly modify one quadrature variable, using ancillary modes,
quadrature measurements, and adaptive control.

While strong deterministic optical nonlinearities at the single-photon level are challenging,
measurement-induced strategies have demonstrated effective nonlinear interactions on weak optical states.
A representative example is the experimental emulation of a strong Kerr nonlinearity via
measurement-induced operations (combining photon addition/subtraction with homodyne-based characterization) \cite{Costanzo2017}.
Such demonstrations support the practical feasibility of measurement-induced nonlinear phase processing of traveling light.

\section{Unclonable encryption}
\label{sec:UE}

%-------------------------------------------------------
\subsection{Background}

We consider schemes that encrypt a classical plaintext into a quantum-state ciphertext.
This is known as {\em Quantum Encryption of Classical Messages} (QECM).
Unclonability of a QECM scheme is defined via a {\em cloning game}.
A challenger prepares the cipherstate for a random key and uniformly drawn plaintext
and gives the cipherstate to Alice.
Alice splits the cipherstate into two pieces; one piece goes to Bob, one to Charlie.
Then Bob and Charlie receive the decryption key and must {\em both} produce the correct plaintext, without
being allowed to communicate.
A QECM scheme is considered to be $\qd$-secure against cloning if the three players ABC, acting together,
have advantage less than $\qd$ of winning the game compared to trivial random guessing.

Recently it has been shown \cite{BC2026} that information-theoretic
unclonability can be achieved in the case of a single-bit message.
The security proof makes use of a decoupling theorem \cite{dupuis2014oneshotdecoupling}:
If one starts from a state $\qr^{A_1 E}$, where A$_1$ and E are potentially entangled, then
mixing the A$_1$-space into a higher-dimensional Hilbert space A$=$A$_1$A$_2$ using the unitaries from
a unitary two-design causes decoupling between A$_1$ and~E.
The `quality' of the decoupling (trace distance between actual and ideal situaton) is quantified as 
$\frac{3\log\log d}{2\log d}$  \cite{BC2026}, where $d$ is the dimension of the A system.
This result can be generalised to the case of decoupling using $\qe$-{\em approximate} unitary two-designs \cite{szehr2013decoupling}, yielding
$\frac{3\log\log d}{2\log d}\sqrt{1+\qe\cdot 4d^4}$ (where $\qe$ is a bound on the diamond norm).

Below we list the most important definitions and lemmas.

\begin{definition}[QECM]
\label{def:QECM}
(\cite{BC2026})
A Quantum Encryption of Classical Messages (QECM) is a tuple $Q=(\cK,\cX,A,\mu,\{\qs^k_x\}_{k\in\cK,x\in\cX})$,
where $\cK$ is the key space; $\cX$ is the message space;
A is a register that holds the encrypted message;
$\mu$ is a probability distribution on $\cK$;
 $\qs^k_x\in\cD(\cH_A)$ is the encryption of message $x$ with key~$k$.
We say that $Q$ is $\qh$-correct if there exists a family of decryption maps
$D^k: \cD(\cH_A)\to \cD(\cH_X)$ such that
$\forall_{k\in\cK, x\in\cX} \bra x D^k(\qs^k_x) \ket x \geq \qh$.
We say that $Q$ is correct if it is $1$-correct.
\end{definition}

\begin{definition}[Cloning attack and $\qd$-unclonable security]
\label{def:CloningAttack}
(\cite{BC2026})

\noindent
Consider a QECM $Q=(\cK, \cX, A, \mu,\{\qs^k_x\}_{k\in\cK,x\in\cX})$.
A {\bf cloning attack} against $Q$
is a tuple $\cA=(B,C,$ $\{B^k_x\}_{k\in\cK,x\in\cX} , \{C^k_x\}_{k\in\cK,x\in\cX},\qF )$, where
B and C are registers representing the system of Bob and Charlie respectively;
$\{B^k_x\}_{x\in\cX}$ and $\{C^k_x\}_{x\in\cX}$ are POVMs on $\cH_B$ and $\cH_C$ representing
Bob and Charlie's measurements given~$k$;
$\qF: \cD(\cH_A)\to \cD(\cH_{BC})$ is a CPTP map representing the cloning channel.
The success probability of the cloning attack against $Q$ is
${\rm c}(Q,\cA)=\EE_{k} \sum_{x\in\cX}\frac1{|\cX|} \Tr [(B^k_x\otimes C^k_x)\qF(\qs^k_x)]$.
We say that the QECM Q is {\bf $\qd$-unclonable secure} if
\be
    \sup_\cA c(Q,\cA) \leq \frac1{|\cX|}+\qd.
\ee
\end{definition}

\begin{definition}[Cloning-distinguishing attack and $\qd$-unclonable-indistinguishable security]
\label{def:UnclDist}
(\cite{BC2026})
Consider a QECM $Q=(\cK, \cX, A, \mu,\{\qs^k_x\}_{k\in\cK,x\in\cX})$.
A {\bf cloning-distinguishing attack} against $Q$ is
a tuple $\cA=(\{x_0,x_1\},B,C,\{B^k_b\}_{k\in\cK,b\in\bits} , \{C^k_b\}_{k\in\cK,b\in\bits},\qF )$,
where 
$x_0,x_1\in\cX$ with $x_1\neq x_0$;
B and C are registers representing the system of Bob and Charlie respectively;
$\{B^k_b\}_{b\in\bits}$ and $\{C^k_b\}_{b\in\bits}$ are POVMs on $\cH_B$ and $\cH_C$ representing
Bob and Charlie's measurements given~$k$;
$\qF: \cD(\cH_A)\to \cD(\cH_{BC})$ is a CPTP map representing the cloning channel.
The success probability of the cloning-distinguishing attack against $Q$ is
${\rm cd}(Q,\cA)=\EE_k \!\sum_{b\in\bits}\frac12\Tr [(B^k_b\otimes C^k_b)\qF(\qs^k_{x_b})]$.
We say that $Q$ is {\bf $\qd$-unclonable-indistinguishable secure}~if 
\be
    \sup_\cA {\rm cd}(Q,\cA) \leq \frac12 +\qd.
\ee
\end{definition}
Unclonable-indistinguishable security implies unclonable security \cite{broadbent_et_al:LIPIcs.TQC.2020.4}.
The two notions are identical if the message space is~$\bits$.
Furthermore, unclonable-indistinguishable security implies indistinguishable security~\cite{broadbent_et_al:LIPIcs.TQC.2020.4}.

\begin{lemma}
\label{lemma:UEfromDecoupl}
(Theorem~5.2 in \cite{BC2026})
Let $\cV\subset \cU(\cH_A)$ be a 2-design.
Let $d=$dim$(\cH_A)$.
The QECM $(\cV, \{0,1\}, A, \mu_\cV, \{U(\ket x\bra x\otimes\qt) U^\dagger\}_{U\in\cV,x\in\{0,1\}})$
is correct and $\qd$-unclonable secure with $\qd=\frac{3\log\log d}{2\log d}$.
\end{lemma}

%-------------------------------------------------------------------
\subsection{Unclonable Encryption scheme based on position and momentum operators}
\label{sec:newUECV}

We propose a QECM for one-bit messages, based on the two-design unitaries from Section~\ref{sec:construction}.
This results in a new Unclonable Encryption scheme that is based on position and momentum operators.
It can be used (at large $d$) to encrypt a classical bit into a 
single-mode CV state.
Below we detail the scheme and prove unclonable security.

Let $\cH_d$ be factorised as $\cH_d=\cH_{\rm sign}\otimes\cH_{\rm dist}$,
with dim$(\cH_{\rm sign})=2$ and dim$(\cH_{\rm dist})=\frac d2$.
A q-quadrature eigenstate $\ket j$ with $j\geq 0$ is decomposed as 
$\ket j= \ket{0}_{\rm sign}\otimes\ket{j}_{\rm dist}$
and for $j<0$ as
$\ket j=\ket{1}_{\rm sign}\otimes\ket{|j|-1}_{\rm dist}$.
Conversely, sign bit $s\in\{0,1\}$ and distance $r\in\{0,\ldots,\frac d2-1\}$
represent q-eigenstate $\ket{(-1)^s (r+s)}$.

The message space is $\{0,1\}$.
The key space is $\cK=[0,d)^{6\ell}$.
The ciphertext space is $\cD(\cH_d)$.
The register containing the cipherstate is called~A.
Let $x\in\bits$ be the plaintext, and $k=(\qa_1\cdots\qa_\ell,\qa_1'\cdots\qa_\ell',$ $\qa_1''\cdots\qa_\ell''
, \qb_1\cdots\qb_\ell,$ $\qb_1'\cdots\qb_\ell',\qb_1''\cdots\qb_\ell'')$ the key.
The key is drawn uniformly from~$\cK$.
We define
\be
    \cV_\ell(d) = \{U_k\}_{k\in\cK}, \quad\quad 
    U_k= (V_{\qa_\ell''\qb_\ell''}\tilde V_{\qa_\ell'\qb_\ell'}V_{\qa_\ell \qb_\ell})\cdots 
    (V_{\qa_1''\qb_1''}\tilde V_{\qa_1'\qb_1'}V_{\qa_1 \qb_1}) 
\ee
% \}_{\qa_i,\qa_i',\qa_i'',\qb_i,\qb_i',\qb_i''\in[0,d)}
with $V_{\qa\qb}=\qo^{\qa\hat Q+\qb\hat Q^2}$, $\tilde V_{\qa\qb}=\qo^{\qa \hat P+\qb\hat P^2}$.
The steps of the encryption are as follows.
\begin{itemize}
\item
Draw uniform $r\in\{0,\ldots,\frac d2-1\}$.
Prepare state $\ket{x}_{\rm sign}\otimes\ket{r}_{\rm dist}$.
Forget~$r$.
The resulting density matrix is $\qr^A_x=\ket x \bra x_{\rm sign}\otimes \qt^{\rm dist}$.
\item
Apply the unitary $U_k$ to create the cipherstate $\qs^A_{kx}=U_k \qr^A_x U_k^\dagger$.
\end{itemize}
Decryption consists of applying $U_k^\dagger$
and then performing a measurement on the $\cH_{\rm sign}$ system.  
Our scheme has the usual structure \cite{MST2021} \cite{BC2026}  
of mixing a plaintext state into a higher-dimensional space. 
Formally, we write our QECM as 
\be
    S_{d\ell} = \Big(\cV_\ell(d), \bits, A, \mu_{\rm unif}, 
    \{U\big(\ket{x}\bra{x}_{\rm sign}\otimes\qt^{\rm dist}\big) U^\dagger\}_{U\in\cV_\ell(d),x\in\{0,1\}} \Big)
\ee
where $\mu_{\rm unif}$ is the uniform distribution.
We note that drawing uniformly from $[0,d)^{6\ell}$ is practical.
\begin{theorem}
\label{th:secure}
The QECM $S_{d\ell}$ is correct and $\qd$-unclonable secure, with
\be
    \qd = \frac{3\log\log d}{2\log d}\sqrt{1+4d^{5-\ell}}.
\ee
\end{theorem}
\begin{proof}
Follows exactly the same steps as the proof of Lemma~\ref{lemma:UEfromDecoupl} (see \cite{BC2026}),
but with the additional factor $\sqrt{1+\qe\cdot 4d^4}$ from the generalised decoupling theorem of 
\cite{szehr2013decoupling}.
In our $\qh$-approximate unitary 2-design we have $\qh=d^{-\ell}$, using the $2\!\to\!2$ norm in the definition.
As \cite{szehr2013decoupling} uses the diamond norm, we have to take into account 
$\qe\leq \qh d$ (Lemma~\ref{lem:dia-2}).
\end{proof}

Note that we need $\ell\geq 5$ in order to get a reasonably small~$\qd$.

%========================================================================
\section{Discussion}
\label{sec:discussion}

The main outcome of this work is that a meaningful notion of an approximate unitary $2$-design can be realised in the continuous-variable setting once one works with a physically motivated regularisation. 
More precisely, we have given an explicit construction of approximate unitary $2$-designs that is compatible with the $p$- and $q$-quadrature structure of CV systems, and whose encryption application yields CV unclonable-indistinguishable security for the first time. 

A central feature of the construction is the discretised shadow of the CV Hilbert space. 
The role of this discretisation is not to replace the CV system by an unrelated finite-dimensional model, but rather to provide a controlled approximation for states of finite energy. 
For a fixed energy cutoff, the discretisation error can be made arbitrarily small by increasing the parameter $d$, so that the discretised description converges to the underlying CV state in a controlled way. 
At the same time, increasing $d$ also improves the quality of the approximate design itself, since the design error scales as $\varepsilon = d^{-\ell}$. 
In this sense, larger $d$ simultaneously sharpens the CV approximation and strengthens the design property.
On the implementation side, 
$d$ determines the bin size and the finite quadrature window on which the boxed phase profiles must be realised. 
The experimentally relevant choice of $d$ should depend on the available state preparation, the energy range of the states being processed, and the accuracy with which the required staircase-like nonlinear phase gates can be approximated. 
Thus the achievable value of $d$ in practice will be set by the experimental platform rather than by the theory alone.

Our construction is motivated by continuous-variable cryptography, but it is also of independent interest as a finite-dimensional scheme. 
Even when viewed purely in finite dimension, without reference to CV implementations, it provides a structured family of approximate unitary $2$-designs that uses fewer parameters than the construction of \cite{Nakata2017Unitary2Designs}. 
Our scheme may also be a natural candidate for implementing approximate $2$-designs in other high-dimensional platforms, such as time-bin systems.

There is also an intriguing connection to mutually unbiased bases (MUBs). 
The operators $\qo^{\qa \hat Q+\qb \hat Q^2}$, for integer $\qa,\qb$, appear in the construction of MUB states in prime dimension. 
This suggests a further strengthening of the known relationship between $2$-designs and MUBs.

For the CV setting, several natural extensions remain open. 
One obvious direction is to investigate whether the $q$-quadrature cutoff can be carried out without the discretisation. 
Such an approach could help clarify to what extent the present discretised model is merely a technical device, and to what extent it captures a more intrinsic CV design structure.

Another point concerns the instantiation of the QECM $S_{d\ell}$ using the approximate unitary $2$-design of Section~\ref{sec:prime}, where $d$ is prime. 
The main advantage of this version is that it yields a discrete key space, $\cI_d^{6\ell}$. 
However, because the QECM requires even dimension, the plaintext state $\ket x\bra x\otimes \qt^{\rm dist}$ must then be prepared in a $(d-1)$-dimensional subspace of $\cH_d$, with the origin excluded. 
The subsequent encryption mixes this state into the full $d$-dimensional space, so the construction deviates slightly from the standard QECM structure. 
Understanding how this modification affects the security proof is left for future work.

More generally, the scheme $S_{d\ell}$ admits several possible embellishments and variants. 
For example, one could investigate encodings in which multiple bits are stored in a single CV mode, or replace the sign degree of freedom by a different binary degree of freedom for representing the plaintext bit. 
Another practically relevant direction would be to devise alternative ways of preparing the state $\qt^{\rm dist}$ that avoid the need to generate a narrow approximate eigenstate of $\hat Q$.

\vskip4mm
\noindent
{\bf\large Acknowledgements}\\
This work was partly supported by the Dutch NGF Quantum Delta NL CAT-2 project.

%%%%%%%%%%%%%%%%%%%%%%%%%%%%%%%%%%%%%%%%%%%%%%%%%%%%%%%%%%%%%%%%%%%%%%%%%%%%%%%%%%%%%
%%%%%%%%%%%%%%%%%%%%%%%%%%%%%%%%%%%%%%%%%%%%%%%%%%%
\newpage

%\bibliographystyle{unsrt}
%\bibliography{biblio}

\printbibliography

@article{Mezzadri2007,
  title = {How to generate random matrices from the classical compact groups},
  author = {F. Mezzadri},
  year = {2007},
  journal = {Notices of the AMS},
  volume = {54},
  number = {5},
  pages = {592-604}
}

@inproceedings{Aar2004,
  title = {{Limitations of quantum advice and one-way communication}},
  author = {Scott Aaronson},
  year = {2004},
  booktitle = {19th IEEE Annual Conference on Computational Complexity},
  publsiher = {IEEE},
  pages = {320-332}
}

@article{Winter1999,
  title = {Coding theorem and strong converse for quantum channels},
  author = {Andreas Winter},
  year = {1999},
  journal = {IEEE Transactions on Information Theory},
  volume = {45},
  number = {7},
  pages = {2481–2485}
}

@misc{MST2021,
  author = {C. Majenz and C. Schaffner and M. Tahmasbi},
  title = {Limitations on uncloneable encryption and simultaneous one-way-to-hiding},
  year = {2021},
  note = {\url{https://arxiv.org/pdf/2103.14510}}
}

@article{BC2026,
  author = {A. Bhattacharyya and E. Culf},
  title = {{Uncloneable encryption from decoupling}},
  year = {2026},
  journal = {Nature Physics},
  volume = {22},
  pages = {315-318}
}

@article{DCEL2009,
  author = {C. Dankert and R. Cleve and J. Emerson and E. Levine},
  title = {{Exact and approximate unitary 2-designs and their application to fidelity estimation}},
  year = {2009},
  journal = {Phys.Rev.A},
  volume = {80},
  pages = {012304},
  note = {\url{https://arxiv.org/abs/quant-ph/0606161}},
  doi = {https://doi.org/10.1103/PhysRevA.80.012304}
}

@article{CLLW2016,
  author = {R. Cleve and D. Leung and L. Liu and C. Wang},
  title = {{Near-linear constructions of exact unitary 2-designs}},
  year = {2016},
  journal = {Quantum Information \& Computation},
  volume = {16},
  number = {9-10},
  pages = {721-756},
}

@article{ZSYY2019,
  author = {Q. Zhuang and T. Schuster and B. Yoshida and N.Y. Yao},
  title = {{Scrambling and complexity in phase space}},
  year = {2019},
  journal = {Phys.Rev.A},
  volume = {99},
  pages = {062334},
  doi = {10.1103/PhysRevA.99.062334}
}

@article{ISGA2024,
  author = {J.T. Iosue and K Sharma and M.J. Gullans and V.V. Albert},
  title = {{Continuous-variable quantum state designs: theory and applications}},
  year = {2024},
  journal = {Phys.Rev.X},
  volume = {14},
  pages = {011013},
  doi = {https://doi.org/10.1103/PhysRevX.14.011013}
}

@article{CIBA2025,
  author = {J. Conrad and J.T. Iosue and  A.G. Burchards and V.V. Albert},
  title = {Continuous-Variable Designs and Design-Based Shadow Tomography from Random Lattices},
  year = {2025},
  journal = {Phys.Rev.Lett.},
  volume = {135},
  pages = {060802},
  doi = {https://doi.org/10.1103/dy4m-gq5c}
}

@misc{Low2010PseudorandonmessAL,
  title={Pseudo-randonmess and Learning in Quantum Computation},
  author={R.A. Low},
  year={2010},
  note = {PhD thesis. \url{https://arxiv.org/abs/1006.5227}}
}

@book{HLP1952,
  author = {G.H. Hardy and J.E Littlewood and G. P\'{o}lya},
  title = {Inequalities (2nd edition)},
  year = {1952},
  publisher = {Cambridge University Press}
}

@misc{poincare,
author = {Bergounioux, Maïtine},
year = {2011},
title = {{On Poincar\'{e}-Wirtinger inequalities in spaces of functions of bounded variation}},
note = {\url{https://hal.science/hal-00515451v1/document}}
}

@article{BlumeKohoutTurner2014Curious,
  author  = {Blume-Kohout, Robin and Turner, Peter S.},
  title   = {The Curious Nonexistence of Gaussian 2-Designs},
  journal = {Communications in Mathematical Physics},
  volume  = {326},
  pages   = {755--771},
  year    = {2014},
  doi     = {10.1007/s00220-014-1894-3}
}

@article{Harrow2009,
  author  = {Harrow, Aram W. and Low, Richard A.},
  title   = {Random Quantum Circuits are Approximate 2-designs},
  journal = {Communications in Mathematical Physics},
  year    = {2009},
  volume  = {291},
  number  = {1},
  pages   = {257--302},
  month   = oct,
  doi     = {10.1007/s00220-009-0873-6},
  url     = {https://doi.org/10.1007/s00220-009-0873-6},
  issn    = {1432-0916}
}

@article{Nakata2017Unitary2Designs,
  author  = {Nakata, Yoshifumi and Hirche, Christoph and Morgan, Ciara and Winter, Andreas},
  title   = {Unitary 2-designs from random X- and Z-diagonal unitaries},
  journal = {Journal of Mathematical Physics},
  volume  = {58},
  pages   = {052203},
  year    = {2017},
  doi     = {10.1063/1.4983266},
  eprint  = {1502.07514},
  archivePrefix = {arXiv},
  primaryClass  = {quant-ph}
}

@InProceedings{broadbent_et_al:LIPIcs.TQC.2020.4,
  author    = {Broadbent, Anne and Lord, S{\'e}bastien},
  title     = {{Uncloneable Quantum Encryption via Oracles}},
  booktitle = {15th Conference on the Theory of Quantum Computation, Communication and Cryptography (TQC 2020)},
  pages     = {4:1--4:22},
  series    = {Leibniz International Proceedings in Informatics (LIPIcs)},
  volume    = {158},
  year      = {2020},
  editor    = {Flammia, Steven T.},
  publisher = {Schloss Dagstuhl -- Leibniz-Zentrum f{\"u}r Informatik},
  address   = {Dagstuhl, Germany},
  doi       = {10.4230/LIPIcs.TQC.2020.4},
  url       = {https://drops.dagstuhl.de/entities/document/10.4230/LIPIcs.TQC.2020.4}
}

@Article{bhattacharyyaCulf2026uncloneable,
  author       = {Bhattacharyya, Archishna and Culf, Eric},
  title        = {Uncloneable encryption from decoupling},
  journal      = {Nature Physics},
  year         = {2026},
  volume       = {22},
  number       = {2},
  pages        = {315--318},
  doi          = {10.1038/s41567-025-03154-7},
  eprint       = {2503.19125},
  archivePrefix= {arXiv},
  primaryClass = {quant-ph},
  note         = {Published online 04 Feb 2026}
}

@Article{dupuis2014oneshotdecoupling,
  author  = {Dupuis, Fr{\'e}d{\'e}ric and Berta, Mario and Wullschleger, J{\"u}rg and Renner, Renato},
  title   = {One-shot decoupling},
  journal = {Communications in Mathematical Physics},
  year    = {2014},
  volume  = {328},
  number  = {1},
  pages   = {251--284},
  doi     = {10.1007/s00220-014-1990-4},
  eprint  = {1012.6044},
  archivePrefix = {arXiv},
  primaryClass  = {quant-ph}
}

@Article{horodecki2005partial,
  author  = {Horodecki, Micha{\l} and Oppenheim, Jonathan and Winter, Andreas},
  title   = {Partial quantum information},
  journal = {Nature},
  year    = {2005},
  volume  = {436},
  number  = {7051},
  pages   = {673--676},
  doi     = {10.1038/nature03909}
}

@Article{horodecki2007merging,
  author  = {Horodecki, Micha{\l} and Oppenheim, Jonathan and Winter, Andreas},
  title   = {Quantum state merging and negative information},
  journal = {Communications in Mathematical Physics},
  year    = {2007},
  volume  = {269},
  number  = {1},
  pages   = {107--136},
  doi     = {10.1007/s00220-006-0118-x},
  eprint  = {quant-ph/0512247},
  archivePrefix = {arXiv}
}

@Article{szehr2013decoupling,
  author  = {Szehr, Oleg and Dupuis, Fr{\'e}d{\'e}ric and Tomamichel, Marco and Renner, Renato},
  title   = {Decoupling with unitary approximate two-designs},
  journal = {New Journal of Physics},
  year    = {2013},
  volume  = {15},
  number  = {5},
  pages   = {053022},
  doi     = {10.1088/1367-2630/15/5/053022},
  eprint  = {1109.4348},
  archivePrefix = {arXiv},
  primaryClass  = {quant-ph}
}

@Misc{raySkoric2025cvue,
  author       = {Ray, Arpan Akash and {\v{S}}kori{\'c}, Boris},
  title        = {Practical Unclonable Encryption with Continuous Variables},
  howpublished = {arXiv preprint},
  year         = {2025},
  eprint       = {2503.02648},
  archivePrefix= {arXiv},
  primaryClass = {quant-ph}
}

@article{Got23,
author = {Gottesman, Daniel},
title = {Uncloneable encryption},
year = {2003},
issue_date = {November 2003},
publisher = {Rinton Press, Incorporated},
address = {Paramus, NJ},
volume = {3},
number = {6},
issn = {1533-7146},
journal = {Quantum Info. Comput.},
month = nov,
pages = {581–602},
numpages = {22}
}

@article{LloydBraunstein1999,
  author        = {Lloyd, Seth and Braunstein, Samuel L.},
  title         = {Quantum Computation over Continuous Variables},
  journal       = {Physical Review Letters},
  volume        = {82},
  number        = {8},
  pages         = {1784--1787},
  year          = {1999},
  doi           = {10.1103/PhysRevLett.82.1784},
  archivePrefix = {arXiv},
  eprint        = {quant-ph/9810082}
}

@article{BraunsteinVanLoock2005,
  author  = {Braunstein, Samuel L. and van Loock, Peter},
  title   = {Quantum information with continuous variables},
  journal = {Reviews of Modern Physics},
  volume  = {77},
  pages   = {513--577},
  year    = {2005},
  doi     = {10.1103/RevModPhys.77.513}
}

@article{BraunsteinKimble1998,
  author  = {Braunstein, Samuel L. and Kimble, H. J.},
  title   = {Teleportation of Continuous Quantum Variables},
  journal = {Physical Review Letters},
  volume  = {80},
  number  = {4},
  pages   = {869--872},
  year    = {1998},
  doi     = {10.1103/PhysRevLett.80.869}
}

@article{Furusawa1998,
  author  = {Furusawa, Akira and S{\o}rensen, Jens Lykke and Braunstein, Samuel L. and Fuchs, Christopher A. and Kimble, H. Jeff and Polzik, Eugene S.},
  title   = {Unconditional Quantum Teleportation},
  journal = {Science},
  volume  = {282},
  number  = {5389},
  pages   = {706--709},
  year    = {1998},
  doi     = {10.1126/science.282.5389.706}
}

@article{Menicucci2006,
  author        = {Menicucci, Nicolas C. and van Loock, Peter and Gu, Mile and Weedbrook, Christian and Ralph, Timothy C. and Nielsen, Michael A.},
  title         = {Universal Quantum Computation with Continuous-Variable Cluster States},
  journal       = {Physical Review Letters},
  volume        = {97},
  number        = {11},
  pages         = {110501},
  year          = {2006},
  doi           = {10.1103/PhysRevLett.97.110501},
  archivePrefix = {arXiv},
  eprint        = {quant-ph/0605198}
}

@article{FilipMarekAndersen2005,
  author  = {Filip, Radim and Marek, Petr and Andersen, Ulrik L.},
  title   = {Measurement-induced continuous-variable quantum interactions},
  journal = {Physical Review A},
  volume  = {71},
  number  = {4},
  pages   = {042308},
  year    = {2005},
  doi     = {10.1103/PhysRevA.71.042308}
}

@article{MarekFilip2011,
  author        = {Marek, Petr and Filip, Radim},
  title         = {Deterministic implementation of weak quantum cubic nonlinearity},
  journal       = {Physical Review A},
  volume        = {84},
  number        = {5},
  pages         = {053802},
  year          = {2011},
  doi           = {10.1103/PhysRevA.84.053802},
  archivePrefix = {arXiv},
  eprint        = {1105.4950}
}

@article{Miyata2016,
  author        = {Miyata, Kazunori and Ogawa, Hisashi and Marek, Petr and Filip, Radim and Yonezawa, Hidehiro and Yoshikawa, Jun-ichi and Furusawa, Akira},
  title         = {Implementation of a quantum cubic gate by adaptive non-Gaussian measurement},
  journal       = {Physical Review A},
  volume        = {93},
  number        = {2},
  pages         = {022301},
  year          = {2016},
  doi           = {10.1103/PhysRevA.93.022301},
  archivePrefix = {arXiv},
  eprint        = {1507.08782}
}

@article{Marek2018,
  author  = {Marek, Petr and Filip, Radim and Ogawa, Hisashi and Sakaguchi, Atsushi and Takeda, Shuntaro and Yoshikawa, Jun-ichi and Furusawa, Akira},
  title   = {General implementation of arbitrary nonlinear quadrature phase gates},
  journal = {Physical Review A},
  volume  = {97},
  number  = {2},
  pages   = {022329},
  year    = {2018},
  doi     = {10.1103/PhysRevA.97.022329}
}

@article{Costanzo2017,
  author        = {Costanzo, Luca S. and Coelho, Antonio S. and Biagi, Nicola and Fiur{\'a}{\v{s}}ek, Jarom{\'i}r and Bellini, Marco and Zavatta, Alessandro},
  title         = {Measurement-Induced Strong Kerr Nonlinearity for Weak Quantum States of Light},
  journal       = {Physical Review Letters},
  volume        = {119},
  number        = {1},
  pages         = {013601},
  year          = {2017},
  doi           = {10.1103/PhysRevLett.119.013601},
  archivePrefix = {arXiv},
  eprint        = {1706.07018}
}

\end{document}